\documentclass[conference]{IEEEtran}
\IEEEoverridecommandlockouts

\usepackage[utf8]{inputenc}
\usepackage[T1]{fontenc}
\usepackage{graphicx}
\usepackage{breqn}
\usepackage{pgfplots}
\usepackage{subcaption}
\usepackage{amsmath}
\usepackage{mathtools}
\DeclarePairedDelimiter\floor{\lfloor}{\rfloor}
\usepackage{amssymb}
\usepackage{amsthm}
\usepackage{url}
\usepackage{tabularx}
\usepackage{bbm}
\usepackage{tikz}
\usepackage{comment}

\usepackage[left=1.62cm,right=1.62cm,top=1.9cm]{geometry}

\usetikzlibrary{automata,arrows,positioning,calc}

\usetikzlibrary{shapes.geometric,backgrounds,fit,arrows}

\usepackage{eucal}

\usepackage[ruled]{algorithm2e}

\usepackage[normalem]{ulem}

\newtheorem{theorem}{Theorem}
\numberwithin{theorem}{subsection}

\newtheorem{definition}[theorem]{Definition}

\usepackage{hyperref}
\usepackage{mathtools}
\usepackage{stackengine}
\newcommand\underparen[1]{\@ifnextchar_{\uphelp{\uparen{#1}}}{\uparen{#1}}}
\makeatother
\def\uphelp#1_#2{\ensurestackMath{\stackunder[1pt]{#1}{\scriptstyle #2}}}
\newcommand\uparen[1]{\setbox0=\hbox{$#1$}\ensurestackMath{%
  \stackunder[0pt]{#1}{\rotatebox{90}{$\left(%
  \rule[\dimexpr-.5\wd0+\dp\strutbox-1.3pt]{0pt}{\wd0}\right.$}}%
}}
\usepackage{amssymb} % for triangleq

\usepackage{xspace}

\newcommand{\keepcomment}{1} % 1 - Keep comments, 0 - Hide comments

\usepackage[normalem]{ulem}

\ifnum\keepcomment=1
	\usepackage[colorinlistoftodos,textsize=scriptsize]{todonotes} % todo
    \newcommand{\stkout}[1]{\ifmmode\text{\sout{\ensuremath{#1}}}\else\sout{#1}\fi}
    
\else
	
	\usepackage[disable]{todonotes} % todo
\fi

\usepackage{paralist} % To use compactitem

\title{
 Multiple Resource Allocation in Multi-Tenant Edge Computing via Sub-modular Optimization
}

\author{
    \IEEEauthorblockN{Ayoub Ben-Ameur, Andrea Araldo, Tijani Chahed}
    \IEEEauthorblockA{SAMOVAR, Telecom SudParis, Institut Polytechnique de Paris, 91120 Palaiseau, France\\
    \{first\_name\}.\{last\_name\}@telecom-sudparis.eu}
}

\usepackage{array}
\newcolumntype{C}[1]{>{\centering\arraybackslash}m{#1}}
\usepackage{pgfplots}
\pgfplotsset{width=8cm,compat=1.9}

\begin{document}

\maketitle

%%%%%%%%%%%%%%%%%%%%%%%%%%%%%%%%%%%%%%%%%%%%%%
%%%%%%%%%%%%%%%%%%%%%%% ABSTRACT %%%%%%%%%%%%%
%%%%%%%%%%%%%%%%%%%%%%%%%%%%%%%%%%%%%%%%%%%%%%
\begin{abstract}
Edge Computing (EC) allows users to access computing resources at the network frontier, which paves the way for deploying delay-sensitive applications such as Mobile Augmented Reality (MAR). Under the EC paradigm, MAR users connect to the EC server, open sessions and send continuously frames to be processed. The EC server sends back virtual information to enhance the human perception of the world by merging it with the real environment. Resource allocation arises as a critical challenge when several MAR Service Providers (SPs) compete for limited resources at the edge of the network.
In this paper, we consider EC in a multi-tenant environment where the resource owner, i.e., the Network Operator (NO), virtualizes the resources and lets SPs run their services using the allocated slice of resources.
Indeed, for MAR applications, we focus on two specific resources: CPU and RAM, deployed in some edge node, e.g., a central office. We study the decision of the NO about how to partition these resources among several SPs.
We model the arrival and service dynamics of users belonging to different SPs using Erlang queuing model and show that under perfect information, the interaction between the NO and SPs can be formulated as a sub-modular maximization problem under multiple Knapsack constraints. To solve the problem, we use an approximation algorithm, guaranteeing a bounded gap with respect to the optimal theoretical solution. Our numerical results show that the proposed algorithm outperforms baseline proportional allocation in terms of the number of sessions accommodated at the edge for each SP.
\end{abstract}

\begin{IEEEkeywords}
Resource allocation, multi-tenant edge computing, mobile augmented reality, multi-dimensional knapsack problem, queuing model.
\end{IEEEkeywords}

%%%%%%%%%%%%%% INTRODUCTION %%%%%%%%%%%%%%%%%%%
\section{Introduction}
Mobile Augmented Reality (MAR) has become one of the most emerging applications, accompanied by the development of mobile devices and wireless communication. In MAR, the human perception of the world can be enhanced by merging virtual information (generated from object detection, classification, or tracking) with the real environment via mobile devices~\cite{Siriwardhana2021}. However, it is difficult for a mobile device to offer the abundant computation and energy required by MAR applications.

%Big Tech players started already to develop their own AR devices since 2013 when Google announced its open beta of Google Glass~\cite{googleglass} followed by the HoloLens~\cite{hololens} and HoloLens 2~\cite{hololens2} from Microsoft in 2016 and 2019, respectively. Meta Group (formerly Facebook) joined the competition with Virtual Reality (VR) headsets: the Oculus Quest~\cite{quest} in 2019 and Oculus Quest 2~\cite{quest2} in 2020. 

While the production of AR/VR dedicated hardware seems very effective to run AR/VR applications properly, it is costly in the sense that only big players can afford producing their own devices. Hence, multi-tenant EC is particularly interesting for all the other players, as it is probably the only way for small or medium AR Service Providers (SPs) to run their applications at the edge of the network.
The development of EC and 5G has eliminated the obstacle to deploying the MAR service. In the concept of EC~\cite{mao2017survey}, computing and storage resources are deployed at the edge of the access network. %Moreover, supported by the ultra-low latency of 5G technology [3], mobile devices can offload complex computation tasks to the EC network without suffering extra delay.%
Several MAR clients on mobile devices can send MAR requests that contain original data captured by sensors and cameras to the Edge Computing (EC) server. Furthermore, dedicated computing hardware (e.g., Graphics Processing Unit (GPU) and Central Processing Unit (CPU)) and software (e.g., computer vision-based algorithms) process these data and then return the results, such as object classification or space coordinate information, to the mobile devices.

The use of the EC for MAR has attracted extensive attention from the research community and industry recently~\cite{erol2018caching, ayoub2021}, which mainly focus on architecture design and deployment. However, scheduling the MAR requests received from several competing MAR clients on one EC server is critical and challenging. We address in this work the issue of resource allocation to competing, heterogeneous SPs in the case of multiple, limited resources at the Edge. We first model the arrivals and service dynamics of the flows using Erlang queuing model. We then formulate the resource allocation problem to each of the SPs using sub-modular maximization under Knapsack constraints. We next propose an implementation of the so-called streaming algorithm to solve the allocation problem, and obtain a $(\frac{1}{1+2d}-\epsilon)$-approximate optimal value, where $d$ is the number of resource types and $\epsilon$ is a controllable error term. We eventually provide numerical results to show that the resulting system performance significantly outperforms baseline resource allocation policies.

The remainder of this paper is organized as follows. In Section~\ref{sec:related-work} we discuss most relevant work related to ours. We introduce in Section~\ref{sec:system-model} our system model. We formulate the sub-modular maximization problem under Knapsack constraints in Section~\ref{sec:sub-modular-max-knapsack} and describe the proposed algorithm to solve it. In Section~\ref{sec:results}, we show our simulation results. We draw conclusions in Section~\ref{sec:concl}.

%%%%%%%%%%%%%% RELATED WORK %%%%%%%%%%%%%%%%%%%
\section{Related Work}
\label{sec:related-work}
%\subsection{Edge Computing for Mobile Augmented Reality}
Recently, much research effort has been made to develop MAR applications under the EC paradigm. In addition to studies on efficient EC architecture design for MAR \cite{Ren2019}, \cite{fernandez2018fog}, in preliminary studies, researchers concentrated on the resource allocation problem in the MAR service \cite{Wenliang2018}, \cite{jia2018delay}. Some researchers began to notice the trade-off between processing latency and accuracy. They aimed to develop acceleration mechanisms to reduce processing latency~\cite{Lane2016} or characterize the relation between computational complexity and the image size~\cite{Yinghui2020}. Based on these studies, the adaptation of the client configuration (image size and frame rate), and the resource allocation scheme were jointly considered in a centralized manner~\cite{Qiang2018}, \cite{Qiang2018-b}. However, in both studies, the researchers ignored the characteristics of dedicated computing devices (i.e., using batch processing to improve the GPU utility) for MAR tasks. Moreover, their solutions centrally controlled each client configuration, which is challenging to apply to a realistic MAR system.

%\subsection{Resource Allocation in Edge Networks}
In~\cite{josilodan2022}, the authors consider an edge computing system under network slicing in which the wireless devices generate latency sensitive computational tasks. 
The allocation of wireless and computing resources to a set of autonomous wireless
devices in an edge computing system is considered in~\cite{slanda2019}.
A main common assumption of the papers above is that user devices submit tasks to the NO. Contention in these works is modeled among user devices.
However, we consider that these models are not appropriate for EC in our vision, since all traffic between devices and service providers is encrypted to maintain confidentiality and the NO does not have control over it. Therefore the contention for resources is, in our vision, between SPs and not between tasks submitted by users. In our assumption, the NO can only decide how to allocate resources among SPs and then users device interact directly with SPs, outside the control of the NO.
In~\cite{danG2022}, authors consider the interplay between latency constrained applications and function-level resource management in EC. A game theoretic model of the interaction between rate adaptive applications and a load balancing operator is developed under a function-oriented pay-as-you-go pricing model.
In our approach, we assume that the NO does not require any payment from the SPs. The NO aims to maximize his own utility by allocating resources to SPs at the edge. In our vision, an important part of MAR providers cannot afford the payment for resources at the edge.

%%%%%%%%%%%%%% SYSTEM MODEL %%%%%%%%%%%%%%%%%%%
\section{System Model and Optimization Problem}
\label{sec:system-model}

\begin{comment}
    \begin{table*}[t]\centering
  \begin{tabular}{c||c}
   \hline
    Parameter & Definition \\ \hline
    $\mathcal{R}$ & Set of resources\\
    $P$ & Number of SPs\\
    $T$ & Number of time slots\\
    $K^r$ & Total capacity of resource $r$ on the edge node \\ 
    $\lambda_p$ & Users arrival rate for SP $p$ \\
    $\mu_p$ & Users service rate for SP $p$ \\
    $z^r_p$ & Amount of resource $r$ required by one user of SP $p$\\
    $B_p$ & Blocking probability of SP $p$\\
    $N_p$ & Number of users of SP $p$ if it has all the resources on the edge node\\
    $U_E$ & Utility perceived by a user who establishes a session directly in the edge node\\
    $U_C$ & Utility perceived by a user who establishes a session directly in remote cloud\\
    $\mathcal{V}_p$ & Set of users of SP $p$ if it has all the resources on the edge node\\
    \hline
    \hline
    Decision Variable & Definition \\ \hline
    $\vec{\boldsymbol{\theta}}$ & Allocation vector\\
    $\theta^r_p$ & Amount of resource $r$ given to SP $p$ on the edge node\\
    $n_p$ &  Maximum number of users of SP $p$ served at the edge given resources $(\theta_p^r)_{r\in\mathcal{R}}$\\
    $\mathcal{S}_p$ & Set of users of SP $p$ served at the edge given resources $(\theta_p^r)_{r\in\mathcal{R}}$\\
    \hline
  \end{tabular}
  \caption{Summary of used notations}
  \label{tab:notation-table}
\end{table*}
\end{comment}

We consider a setting with one NO, owning a set of resources $\mathcal{R}$ and willing to share them between $P$ different SPs. 
Each SP can then use its assigned share as if it had a dedicated hardware deployed in the edge.

%Table~\ref{tab:notation-table} summarizes the frequently used notations in the paper.

\subsection{Request Pattern}
\label{sec:Request-patterm}
MAR users of SP $p$ arrive to the EC server following a Poisson process with rate $\lambda_p$ expressed in $users/s$. Once a user of any SP $p$ is connected to the edge server, a session is created. This session is valid for a period of time denoted $T_p$ during which the user can perform a sequence of interactions within that MAR application. A single session can contain multiple activities all of which are stored in the session temporarily while the user is connected. Each SP runs in the edge a virtual server, e.g., a Kubernetes POD~\cite{k8s-pods}. A MAR user establishes a session with the virtual server of the respective SP. Within that session, it sends a stream of image processing requests. When users point their MAR device toward an object, raw video from the MAR device cameras are fetched and clipsed into frames with specific image format, such as JPEG and PNG and sent to the edge server~\cite{Ren2019}. The video frames are delivered to the AR tracker to determine the user’s position with respect to the physical surroundings. Given the tracking results, virtual coordinate of the environment can be established by the mapper. Then, the internal objects in video frames are identified by the object recognizer with robust features. The MAR device finally downloads information about the object from the edge server. The AR information is presented in a 3-D “experience” superimposed on the object. What users see, then, is part real and part virtual. Since MAR needs high data rates, ultra-low latency and the possible use of lightweight devices, performing processing at the edge of 5G mobile networks can help guarantee the requirements of MAR applications (Section III-F of~\cite{Siriwardhana2021}).

We assume that a session of a single user of SP $p$ requires a certain amount of resource $r$ denoted $z_p^r$. If the SP does not have at the edge such amount of resources available, the user will establish a session with the cloud, suffering longer delay. Once a user of SP $p$ is served by the edge, his session will be closed and he leaves the EC system. Please note that users can leave the system when they decide, this does not deny that we can define an average service rate for SP $p$ expressed in $users/s$ denoted by $\mu_p = \frac{1}{T_p}$.

\subsection{Resources Partitioning}
\label{sec:resources-partition}

The NO owns CPU and RAM at the edge of the network, for instance, in a server co-located with a (micro) base station or central offices at the metropolitan scale.
It allocates a total capacity $K^{\text{CPU}}$ of CPU and a total capacity $K^{\text{RAM}}$ of RAM among the $P$ SPs. The allocation is a vector $\vec{\boldsymbol{\theta}}$~=~($\vec{\boldsymbol{\theta}^{\text{CPU}}}$,~$\vec{\boldsymbol{\theta}^{\text{RAM}}}$) where each vector $\vec{\boldsymbol{\theta}^r}$ is the allocation of resource $r$. More precisely, the allocation has a form as follows:
\begin{align}
    \vec{\boldsymbol{\theta}} = (\theta^{\text{CPU}}_1, \dots, \theta^{\text{CPU}}_P, \theta^{\text{RAM}}_1, \dots, \theta^{\text{RAM}}_P)
\end{align}

We define the set of all possible allocations as:
\begin{align}
    \label{eq:T}
    \mathcal{T}\triangleq
    \left\{
    \vec{\boldsymbol{\theta}} | \sum_{p=1}^P \theta^r_p\le K^r, \theta^r_p\in\mathbb{Z}^+, r\in\mathcal{R} 
    \right\}
\end{align}

\subsection{Service Model}
\label{sec:service-model}
We model our system as an Erlang queue~\cite{kleinrock1975theory} which models Poisson arrivals, exponentially distributed service time, and a number of servers equal to the number of places in the system, i.e., users are either directly served at the edge or directed to the cloud. In our case, users of SP $p$ arrive to the edge according to a Poisson distribution with mean arrival rate $\lambda_p$, they remain in the system for an exponentially distributed duration, $T_p$. The number of servers in our case refers to the maximum number of sessions that the edge can accommodate for each SP, as determined next. Each user of SP $p$ has fixed requirements $(z_p^r)_{p=1..P, r \in \mathcal{R}}$ and fixed allocation $((\theta_p^r)_{r\in\mathcal{R}})$ during service. We denote by $n_p(\vec{\pmb{\theta}})$ the maximum number of users that can be served at the edge for a SP $p$ when the resource allocation decided by the NO is $\vec{\pmb{\theta}}=(\theta^r_p)_{r\in \mathcal{R}}$. Each user of each SP $p$ will receive an amount $z^r_p$ of the resource $r$ for their session. Hence the maximum number of sessions $n_p(\vec{\pmb{\theta}})$ each SP $p$ can establish at the edge when the allocation from the NO is $\vec{\pmb{\theta}}$ must satisfy:
\begin{align}
n_p(\vec{\pmb{\theta}})\cdot z_p^r \leq \theta_p^r, p = 1\dots P , r\in\mathcal{R}.
    \label{eq:inequality-np}
\end{align}

Therefore, 
${n_p}(\vec{\pmb{\theta}})$ is:
\begin{align}
    n_p(\vec{\pmb{\theta}}) =  \floor*{\min_{r\in\mathcal{R}} \left(\frac{\theta_p^r}{z_p^r} \right)}, p = 1\dots P
\end{align} 
where $\floor*{.}$ is the floor function giving as output the greatest integer less than or equal to $\left(\frac{\theta_p^r}{z_p^r} \right)$.

Let us denote by $N_p$ the number of users of SP $p$ served at the edge if all the resources are allocated only to this SP $p$. 
\begin{align}
    N_p =  \floor*{\min_{r\in\mathcal{R}} \left(\frac{K^r}{z_p^r}\right)}, p = 1\dots P
\end{align}

\subsection{Utility Model}
\label{sec:utility-Model}
A user of SP $p$ is served directly by the edge if the latter can satisfy the requirements $z_p^{\text{RAM}}$ and $z_p^{\text{CPU}}$. Otherwise, the corresponding session is not accepted (we say that it is ``blocked'', following the terminology from queuing theory) and directed to a remote cloud server. 
Using Erlang (equation (3.45) of~\cite{kleinrock1975theory}), the probability for a user of SP $p$ to be blocked is
\begin{align}
B_p(\vec{\boldsymbol{\theta}}) = \frac{\frac{A^{n_p(\vec{\boldsymbol{\theta}})}_p}{n_p(\vec{\boldsymbol{\theta}})!}}{\sum_{i=0}^{n_p(\vec{\boldsymbol{\theta}})} \frac{A^i_p}{i!}}, p = 1\dots P
\label{eq:blocking-probability}    
\end{align}
where $A_p = \frac{\lambda_p}{\mu_p}$. The probability for a user of SP $p$ to have his/her session established with the edge is thus:
\begin{align}
    \bar B_p(\vec{\boldsymbol{\theta}})=1-B_p(\vec{\boldsymbol{\theta}}).
\end{align}

The utility perceived by a user who establishes a session directly in the edge is $U_E$, while if the session is with the cloud, the utility is $U_C$. Such utilities take into account the impact on the Quality of Experience (QoE) of the delay to process every user request, accounting for a larger delay to reach the cloud. Hence, $U_E > U_C > 0$.
For simplicity, we assume that $U_E$ and $U_C$ are the same for all SPs.  Since $1-B_p$ indicates the fraction of users of SP $p$ establishing sessions with the edge, the expected value of the utility perceived by a user of SP $p$ is, by the theorem of total probability:
\begin{equation}
\begin{aligned}
    \mathbb{E} U_p(\vec{\boldsymbol{\theta}}) 
    & = \mathbb{P}(\text{session established with the edge})\cdot U_E \\
    & + 
    \mathbb{P}(\text{session established with the cloud})\cdot U_C\\
    & = \bar B_p(\vec{\boldsymbol{\theta}}) \cdot U_E + (1-\bar B_p(\vec{\boldsymbol{\theta}})) \cdot U_C\\
    & = (U_E-U_C)\cdot \bar B_p(\vec{\boldsymbol{\theta}}) + U_C
\label{eq:utility-function}
\end{aligned}
\end{equation}

By the theorem of total expectation, the utility perceived by a generic user is
\begin{equation}
\begin{aligned}
\mathbb{E}U(\vec{\boldsymbol{\theta}}) &=
    \sum_{p=1}^p 
    \mathbb{E} U_p(\vec{\boldsymbol{\theta}}) \cdot \mathbb{P}(\text{new user is for SP}p)\\
    & = \sum_{p=1}^p w_p\cdot \mathbb{E}U_p(\vec{\boldsymbol{\theta}})
\label{eq:utility-function-any-user}
\end{aligned}
\end{equation}
\noindent where $w_p = \frac{\lambda_p}{\sum_{p'} \lambda_{p'}}$.

\subsection{Optimization Problem}
\label{sec:Optimization-Problem}

The NO aims to maximize the expected value of the utility perceived by a generic user:
\begin{equation}
\begin{aligned}
    \max_{\vec{\pmb{\theta}}} \quad & \mathbb{E}U(\vec{\boldsymbol{\theta}}) \\
    \textrm{s.t.} \quad & \sum_{p=1}^P \theta_p^r \leq K^r, \forall r\in\mathcal{R}
\label{eq:opt-prob-expected-value0}  
\end{aligned}
\end{equation}

Replacing $\mathbb{E}U_p(\vec{\boldsymbol{\theta}})$ with its value found in~(\ref{eq:utility-function}) and observing that $(U_E-U_C)$ and $U_C$ are positive constants, the optimization problem becomes:
\begin{equation}
\begin{aligned}
    \max_{\vec{\pmb{\theta}}} \quad &
    \sum_{p=1}^P
    w_p \bar B_p(\vec{\boldsymbol{\theta}}) \\
    \textrm{s.t.} \quad & \sum_{p=1}^P \theta_p^r \leq K^r, \forall r\in\mathcal{R}
\label{eq:opt-prob-12}    
\end{aligned}
\end{equation}

Thanks to~\eqref{eq:inequality-np} and~\eqref{eq:blocking-probability}, we can express the problem in terms of $\vec{\mathbf{n}}=(n_1,\dots,n_P)$ instead of $\vec{\pmb{\theta}}$:
\begin{equation}
\begin{aligned}
    \max_{\vec{\pmb{n}}} \quad &
    f(\vec{\pmb{n}})=
    \sum_{p=1}^P
    w_p \bar B_p(\vec{\boldsymbol{n}}) \\
    \textrm{s.t.} \quad & 
    \sum_{p=1}^P n_p \cdot z_p^r \leq K^r, \forall r\in\mathcal{R}
\label{eq:opt-prob-2}    
\end{aligned}
\end{equation}
\begin{align}
\text{where}& \quad
\bar B_p(\vec{\boldsymbol{n}})
\triangleq 
1-
\frac{\frac{A^{n_p}_p}{n_p!}}{\sum_{i=0}^{n_p} \frac{A^i_p}{i!}}, p = 1\dots P
\label{eq:bar-Bp}
\end{align}

Observe that $f(\vec{\mathbf{n}})$ is the probability for a generic user to be served with a session with the edge node. This shows that improving the expected user utility~\eqref{eq:opt-prob-expected-value0} is equivalent to maximizing the probability of establishing a session with the edge~\eqref{eq:opt-prob-2}.

\section{Sub-modular Optimization}
\label{sec:sub-modular-max-knapsack}
To describe our problem~\eqref{eq:opt-prob-2} in terms of sub-modular optimization, we interpret a user session established with the edge node as an item. 
Let $\mathcal{V}_p = \{ 1,2, \dots, N_p \}$ be the set of \emph{candidate sessions} of SP $p$ that could coexist in the edge if all resources were given to this SP $p$. Since in reality resources at the edge are not given to one SP only, we need to choose a subset of sessions $\mathcal{S}_p\subseteq\mathcal{V}_p$ to allocate to each SP $p$. This choice induces a certain probability of establishing a session with the edge:
\begin{align}
\bar B_p(\mathcal{S}_p)=
1- \frac{\frac{A^{|\mathcal{S}_p|}_p}{|\mathcal{S}_p|!}}{\sum_{i=0}^{|\mathcal{S}_p|} \frac{A^i_p}{i!}}
\label{eq:barBp-set}
\end{align}
With slight abuse of notation, in the formula above we use the notation $\bar B_p(\cdot)$ as in~\eqref{eq:bar-Bp}, to emphasize that the two quantities are conceptually the same thing, by setting $n_p=|\mathcal{S}_p|$.
Let $\mathcal{V}\triangleq\bigcup_{p=1}^P \mathcal{V}_p$ the set of all candidate sessions and $\mathcal{S}=\bigcup_{p=1}^P\mathcal{S}_p\subseteq\mathcal{V}$ the set of sessions allocated. Set $\mathcal{S}$ is our decision variable. For each SP $p$, we define a non-negative set function $f_p$, taking as input all possible subsets $\mathcal{S}$ of $\mathcal{V}$, as follows:
\[
    f_p(\mathcal{S}) \triangleq
    w_p \cdot \bar B_p(\mathcal{S}\cap\mathcal{V}_p)
    \in[0,1]
\]

\noindent Function $f_p$ represents the probability, for a user that arrives, to be of SP $p$ and to be served with a session at the edge.
We define $f(\mathcal{S})\triangleq \sum_{p=1}^P f_p(\mathcal{S})$. It indicates, for any arriving user, the probability to be served with a session at the edge.

For any subset $\mathcal{S}$ of $\mathcal{V}$, we denote the characteristic vector of $\mathcal{S}$ by $\boldsymbol{x}_{\mathcal{S}}=(x_{\mathcal{S}_1,1},\dots,x_{\mathcal{S}_1,N_1},\dots,x_{\mathcal{S}_P,1},\dots,x_{\mathcal{S}_P,N_P})^T$, where for any $j~\in~[1, N_p]$ and $p = 1,...,P$:
$$
x_{\mathcal{S}_p,j}=\begin{cases}
			1, & \text{if the } j\text{-th item of }\mathcal{V}_p \text{ is in } \mathcal{S}_p \\
            0, & \text{otherwise}
		 \end{cases}
$$

%Let $f~:~\sum_{p=1}^{P}w_p~(1~-~B_p(\mathcal{S}_p))~\rightarrow~[0,\infty)$ be a non-negative set function on the subset $\mathcal{S} = \cup_{p=1}^P \mathcal{S}_p$, of size $\sum_{p=1}^P n_p$, of $\mathcal{V} = \cup_{p=1}^P \mathcal{V}_p$. 

For $\mathcal{S} \subseteq \mathcal{V}$ and $v \in \mathcal{V}$, the marginal gain in $f$ when adding $v$ to set $\mathcal{S}$ is defined as $\Delta_{f}(v|\mathcal{S})~\triangleq~f(\mathcal{S}~\cup~\{v\})~-~f(\mathcal{S})$.
%, which quantifies the increase in $f(\mathcal{S})$ when $v$ is added into subset $\mathcal{S}$. 

We introduce now the $d$-knapsack constraint where $d~=~|\mathcal{R}|$. Let $\boldsymbol{k} = (K^1, \dots, K^d)^T$ be the resource capacity vector and $\boldsymbol{Z}_p~=~(z^r_{p,j})$ denote a $d \times N_p$ matrix, whose $(r,j)$-th entry $z^r_{p,j} > 0$ is the weight of the $j$-th item of $\mathcal{V}_p$ in terms of resource $r$. Since we have assumed (\S\ref{sec:service-model}) that all users of a SP $p$  require the same amount of each resource, $z^r_{p,j}=z_p^r$ for all the items in $\mathcal{V}_p$.
Therefore, the constraint in~\eqref{eq:opt-prob-2} can be expressed by $\boldsymbol{Z}\cdot\boldsymbol{x}_{\mathcal{S}}\leq\boldsymbol{k}$, where $\boldsymbol{Z} = (\boldsymbol{Z}_1,\dots, \boldsymbol{Z}_P)\in\mathbb{R}^{d\times \sum_p N_p}$ and $\boldsymbol{x}_{\mathcal{S}}~\in~\{0,1\}^{\sum_p N_p\times~1}$. 
Problem~\eqref{eq:opt-prob-2} becomes:
\begin{equation}
\begin{aligned}
    \max_{\mathcal{S}} \quad & f(\mathcal{S}) = \sum_{p=1}^P f_p(\mathcal{S}_p)\\
    \textrm{s.t.} \quad & \boldsymbol{Z}\boldsymbol{X}_{\mathcal{S}} \leq \boldsymbol{k}\\
\label{eq:opt-prob-sub-mod}
\end{aligned}
\end{equation}

Without loss of generality, for $1~\leq~i~\leq~d, 1~\leq~j~\leq~N$, we assume that $z_p^r \leq K^r$. That is, no item has a larger weight than the corresponding knapsack budget, since
otherwise such an item would never be selected into $\mathcal{S}$.

We are now ready to study the properties of formulation~(\ref{eq:opt-prob-sub-mod}).
To do so, we recall two common definitions from set-function theory~\cite{fujishige2005submodular}.

\begin{algorithm}[h]

\caption{Streaming Algorithm for sub-modular maximization problem under Knapsack constraints}
\label{alg:streaming}
\KwData{$d, z_p^r, K^r, \lambda_p, \mu_p$}
\KwResult{$\mathcal{S}^*$}
$m \gets 0$\;
$\mathcal{Q} \gets \{ [1+(1+2d)\epsilon]^l | l\in \mathbb{Z} \} $\;
\For{$v \in \mathcal{Q}$}{$\mathcal{S}_v \gets \emptyset$\;
                \For{$1 \leq i \leq d$}{$m \gets \max\{m, f(\{j\})/z_{i,j} \}$\;}
                $\mathcal{Q} \gets \{ [1+(1+2d)\epsilon]^l | l\in \mathbb{Z},\newline \frac{m}{1+(1+2d)\epsilon}\leq [1+(1+2d)\epsilon]^l \leq 2Km \}$\;
                \For{$1 \leq j \leq n$}{\If{$\exists i\in[1,d], z_{i,j} \geq \frac{K}{2} \textbf{ and } \frac{f(\{j\})}{z_{i,j}} \geq \frac{2v}{K^{(1+2d)}}$}{$\mathcal{S}_v \gets \{j\}$\;
                $break$\;}
                \If{$ \forall i\in[1,d], \sum_{l\in \mathcal{S}\cup \{j\} } z_{i,l} \leq K \textbf{ and } \frac{\Delta_f(j|\mathcal{S})}{z_{i,j}} \geq \frac{2v}{K^{(1+2d)}}$}{$\mathcal{S}_v \gets \mathcal{S}_v \cup \{j\}$\;}
                                    }
                }
$\mathcal{S}^* \gets \arg \max_{\mathcal{S}_v, v\in \mathcal{Q}} f(\mathcal{S}_v)$\;

\end{algorithm}

\begin{definition}
A function $f$ is sub-modular if it satisfies that $\Delta_f(v|\mathcal{B})~\leq~\Delta_f(v|\mathcal{A})$, for any $\mathcal{A} \subseteq \mathcal{B} \subseteq \mathcal{V}$ and $v \in \mathcal{V} \setminus \mathcal{B}$. 
\end{definition}

\begin{definition}
A function $f$ is monotone if for any $\mathcal{S} \subseteq \mathcal{V}$ and $v \in \mathcal{V}$, $\Delta_f(v|\mathcal{S}) \geq 0$. 
\end{definition}

\begin{theorem}
    Function $f$ in~\eqref{eq:opt-prob-sub-mod} is monotone and sub-modular.
\end{theorem}

\begin{proof}
Let $\mathcal{S}\subseteq\mathcal{V}$ and $v\in\mathcal{V}$. Suppose in particular that $v~\in~\mathcal{V}_{p'}$.

\begin{equation*}
\begin{aligned}
    \Delta_f(v|\mathcal{S}) = & f(\mathcal{S}\cup\{v\})-f(\mathcal{S})\\
    = & \sum_{p\neq p'} f_p(\mathcal{S}_p) + f_{p'}(\mathcal{S}_{p'}\cup\{v\}) - \sum_{p=1}^P f_p(\mathcal{S}_p)\\
    = & f_{p'}(\mathcal{S}_{p'}\cup\{v\}) - f_{p'}(\mathcal{S}_{p'})\\
    = & w_{p'} \cdot \bar B_{p'}(\mathcal{S}_{p'}\cup\{v\}) - w_{p'} \cdot  \bar B_{p'}(\mathcal{S}_{p'})\\
    \geq & 0,
\end{aligned}
\end{equation*}
\noindent where the last inequality can be obtained by simple calculus from~\eqref{eq:barBp-set}.
This shows that function $f$ is monotone. 

Let us consider sets $\mathcal{A} \subseteq \mathcal{B} \subseteq \mathcal{V}$ and a vector $v \in \mathcal{V} \setminus \mathcal{B}$.

\begin{equation*}
\begin{aligned}
    \Delta_f(v|\mathcal{B}) - \Delta_f(v|\mathcal{A}) = & [f(\mathcal{B}\cup\{v\})-f(\mathcal{B})] - \\
    & [f(\mathcal{A}\cup\{v\})-f(\mathcal{A})] \\
    = & [f(\mathcal{B}\cup\{v\}) - f(\mathcal{A}\cup\{v\})] + \\
    & [f(\mathcal{A}) - f(\mathcal{B})]\\
\end{aligned}
\end{equation*}

Having $\mathcal{A} \subseteq \mathcal{B}$, we can write $\exists \mathcal{Q}\subseteq\mathcal{V}/\mathcal{B}=\mathcal{A}\cup\mathcal{Q}$. Hence:

\begin{equation*}
\centering
\begin{aligned}
    & [f(\mathcal{B}\cup\{v\}) - f(\mathcal{A}\cup\{v\})] + [f(\mathcal{A}) - f(\mathcal{B})] && \\
     & = [f(\mathcal{A}\cup\mathcal{Q}\cup\{v\}) - f(\mathcal{A}\cup\{v\})] + [f(\mathcal{A}) - f(\mathcal{A}\cup\mathcal{Q})] && \\
     & \leq [f(\mathcal{A})+f(\mathcal{Q}\cup\{v\}) - f(\mathcal{A})-f(\{v\})] && \\
     & + [f(\mathcal{A}) - f(\mathcal{A}\cup\mathcal{Q})] &&\\
     & \leq [f(\mathcal{Q}\cup\{v\}) - f(\{v\})] + [f(\mathcal{A}) - f(\mathcal{A}) - f(\mathcal{Q})]&& \\
     & = f(\mathcal{Q}\cup\{v\}) - [f(\{v\}) + f(\mathcal{Q})]\leq 0 &&
\end{aligned}
\end{equation*}

Therefore, the function $f$ is sub-modular. 
\end{proof}

Now that we have proved that our objective function $f$ is monotone and sub-modular, we can use well known results from sub-modular optimization. IN particular, we adopt the algorithms proposed in~\cite{Qilian2016}, which we report in Algorithm~\ref{alg:streaming}. The main idea of the algorithm is for every potential new user for each SP $p$, we compare the increase in $f$ when we add this user to the set of users $\mathcal{S}$. We add the user providing the most increase in $f$. 
The algorithm guarantees the following sub-optimality gap (Theorem 1 of \cite{Qilian2016}).
\begin{theorem}
    Algorithm~\ref{alg:streaming} outputs $\mathcal{S}$ that satisfies $f(\mathcal{S}) \geq (\frac{1}{1+2d}-\epsilon) OPT$ and has $O(\frac{\log(K_{\max})}{\epsilon})$ computational complexity per element, $d$ being the number of resources, $0 < \epsilon < \frac{1}{1+2d}$, $K_{\max}=\max_{1\leq i \leq d}K^i$ and OPT the value of $f$ obtained by the optimal solution.
\label{th:optimality-gap}
\end{theorem}

Note that the hyper-parameter $\epsilon$ impacts the behavior of the algorithm as well as the quality of the optimality gap. The smaller is $\epsilon$, the larger is our $f(\mathcal{S})$.

\section{Numerical Results}
\label{sec:results}

We now evaluate the performance of Algorithm~\ref{alg:streaming} via a numerical model developed in Python and compare it to the proportional allocation where $\theta_p^r$ is proportional to the arrival rate $\lambda_p$ of users of each SP $p$. We set $\epsilon=0.01$.

\subsection{Setting}
We focus on an edge node co-located with a central offices serving 2 SPs. We set arrival rates $\lambda_1$ and $\lambda_2$ at 20 and 5 $users/s$, respectively and departure rates $\mu_1$ and $\mu_2$ at 1 and 10 $users/s$, respectively.
Motivated by Amazon EC2 instances, such as G4dn~\cite{amazon-ec2}, designed to support machine learning inference for applications like adding metadata to an image, object detection, recommendation systems, automated speech recognition, and language translation, we consider an edge server similar to the G4dn.metal with $K^{\text{RAM}} = 384$ GB of total RAM capacity and a 2nd Generation Intel Xeon Scalable CPU: Cascade Lake P-8259L with total capacity of CPU $K^{\text{CPU}} = 96$ vCPU. Taking in consideration AR applications similar to Pokemon GO~\cite{pokemon-go}, we set RAM and CPU requirements for SP 1 and SP 2 at: $z_1^{\text{RAM}} = 2$ GB, $z_1^{\text{CPU}} = 1$ vCPU, $z_2^{\text{RAM}} = 0.5$ GB and $z_2^{\text{CPU}} = 4$ vCPU, respectively.

\begin{figure}[t]
\begin{subfigure}{.24\textwidth}
  \centering
  \includegraphics[width=1\linewidth]{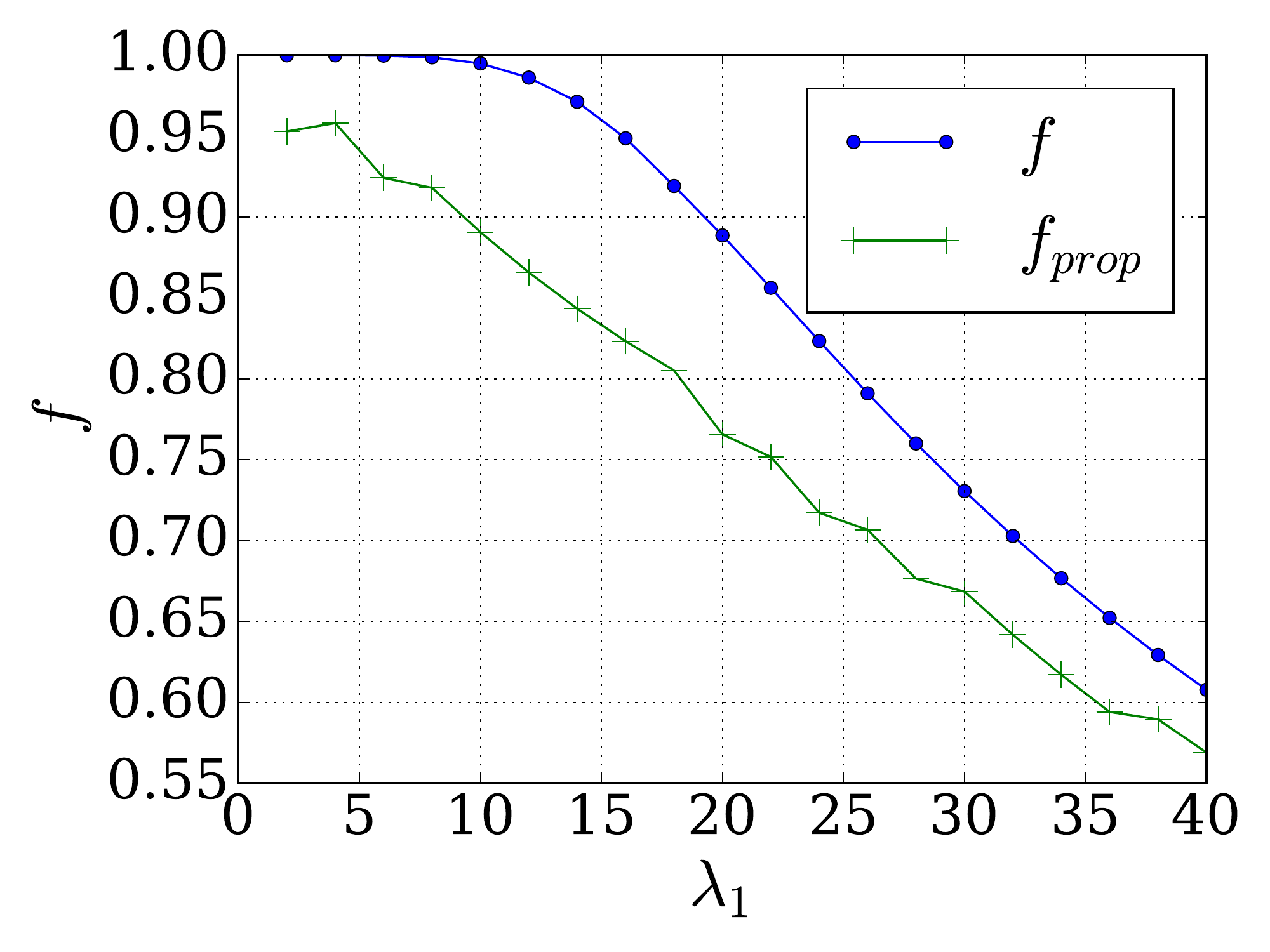}
  \caption{Objective function $f$ vs. $\lambda_1$}
  \label{fig:f}
\end{subfigure}%
\begin{subfigure}{.24\textwidth}
  \centering
  \includegraphics[width=1\linewidth]{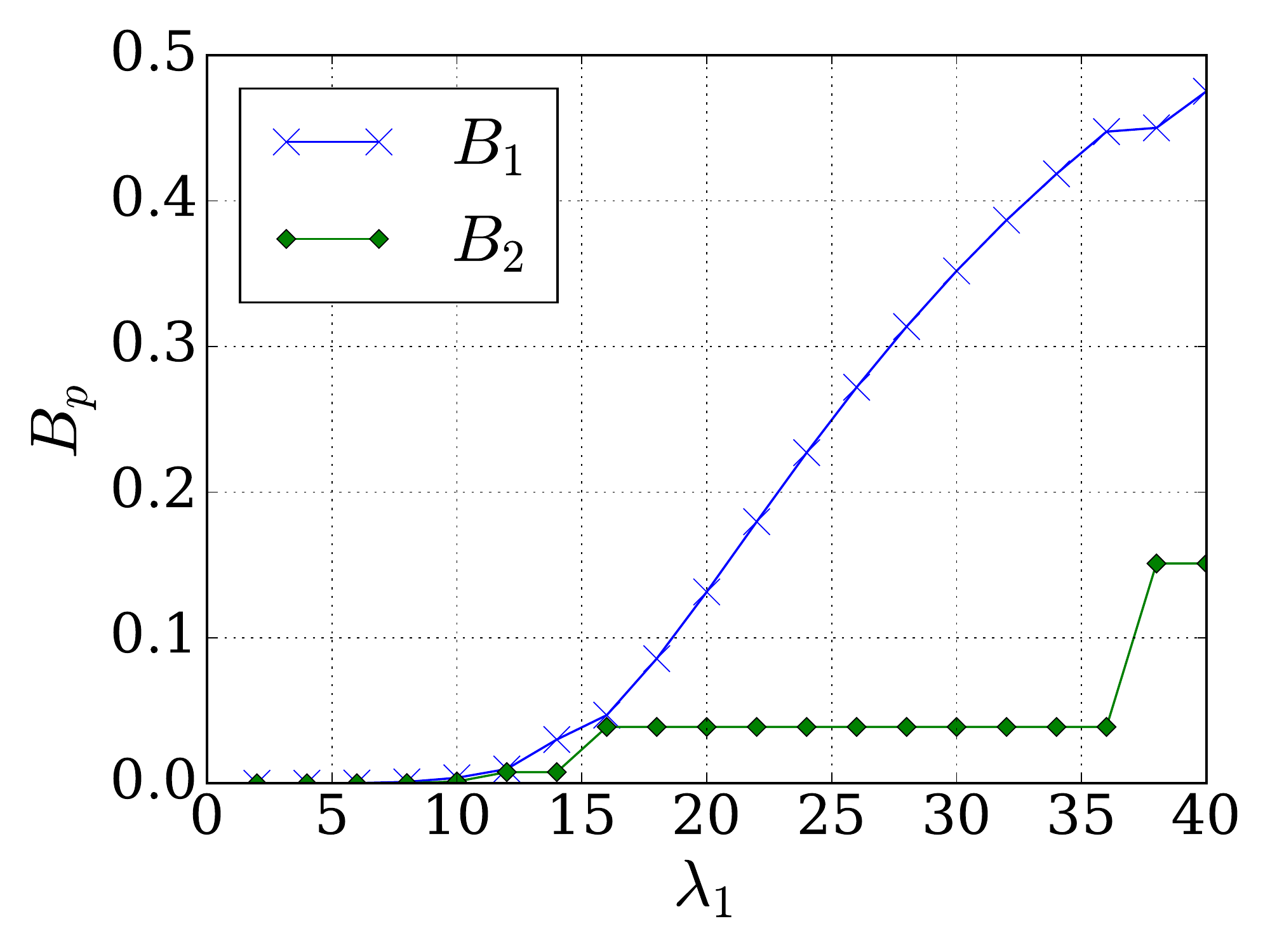}
  \caption{Blocking probability $B_p$ vs. $\lambda_1$}
  \label{fig:B-p}
\end{subfigure}
\caption{Performance of the streaming algorithm  w.r.t $\lambda_1$}
\label{fig:performance}
\end{figure}

\subsection{Results}
We plot in Fig.~\ref{fig:f} our solution obtained with Algorithm~\ref{alg:streaming}: the objective function $f$, which is the probability for a user to establish a session with the edge~\eqref{eq:opt-prob-2} and we compare our solution with the baseline $f_{\text{prop}}$, i.e., the probability of establishing sessions with the edge obtained when allocating resources to SPs proportionally to their users arrival rates. 
In Fig.~\ref{fig:B-p}, we show the variation of the blocking probabilities for each SP when varying $\lambda_1$. The increase in $\lambda_1$ results higher blocking probability for SP 1, which is expected as more users will consume more resources at the edge and less resources are left. Higher $\lambda_1$ will also affect SP 2 but much less significantly. 
As for resource utilization, we plot Fig.~\ref{fig:alloc}. The results show that the CPU is totally utilized by the two SPs (Fig.~\ref{fig:cpu-alloc}), while the RAM is not fully exploited (less than 20\% as shown in Fig.~\ref{fig:ram-alloc}). Despite having more than 80\% of RAM free, we cannot expect better performance since the blocking comes always from the CPU, which is the scarcer resource. Having higher arrival rate, the algorithm does not allow yet SP~1 to have more CPU as this resource is almost 80\% used by SP~2. We can explain this by looking to the values of $z_1^{\text{CPU}}$ and $z_2^{\text{CPU}}$, we can see that SP 2 is CPU-greedy: users of SP 2 consume 4 times more CPU than users of SP 1.

\begin{figure}[t]
\begin{subfigure}{.24\textwidth}
  \centering
  \includegraphics[width=1\linewidth]{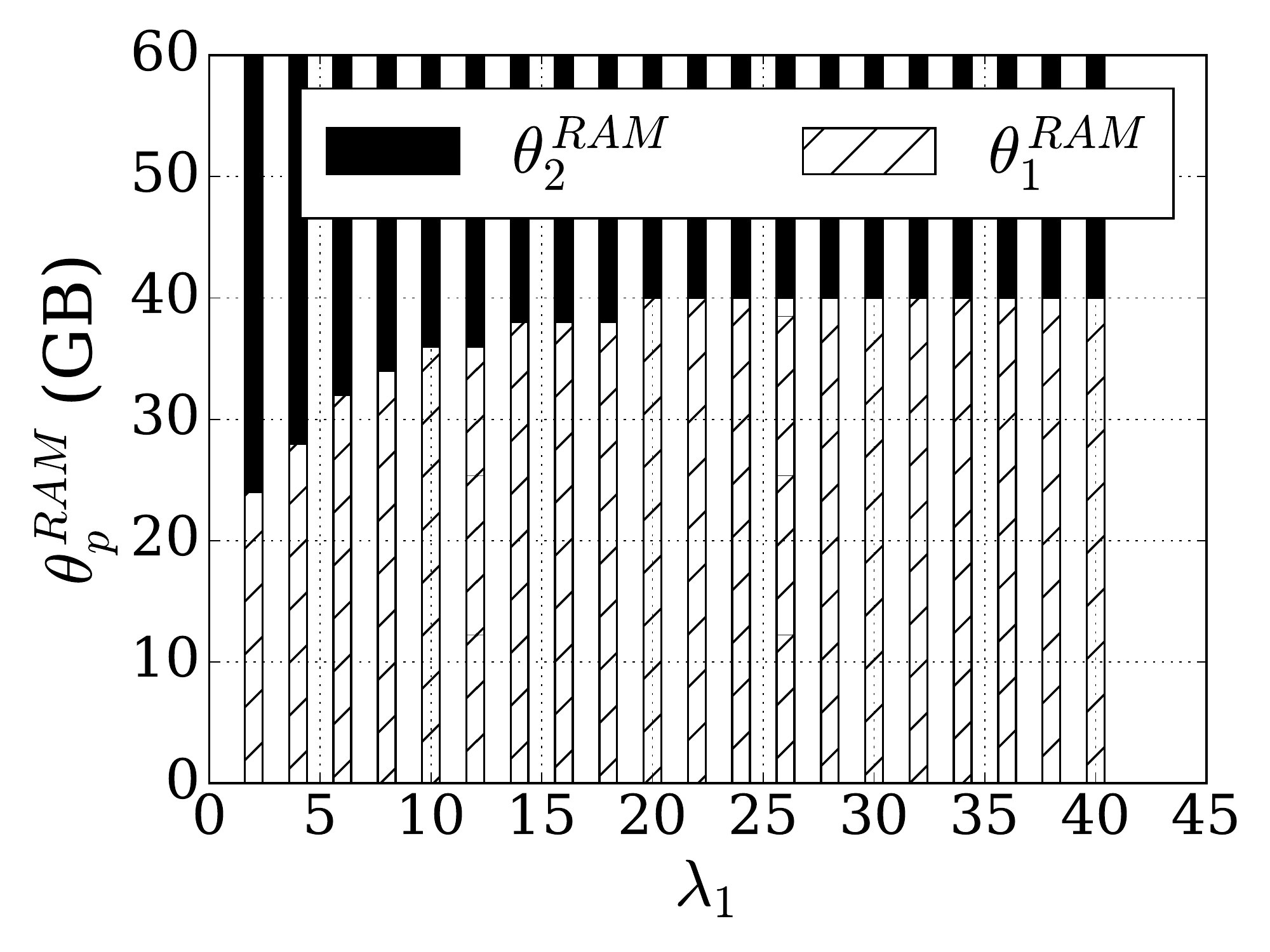}
  \caption{RAM allocation vs. $\lambda_1$}
  \label{fig:ram-alloc}
\end{subfigure}%
\begin{subfigure}{.24\textwidth}
  \centering
  \includegraphics[width=1\linewidth]{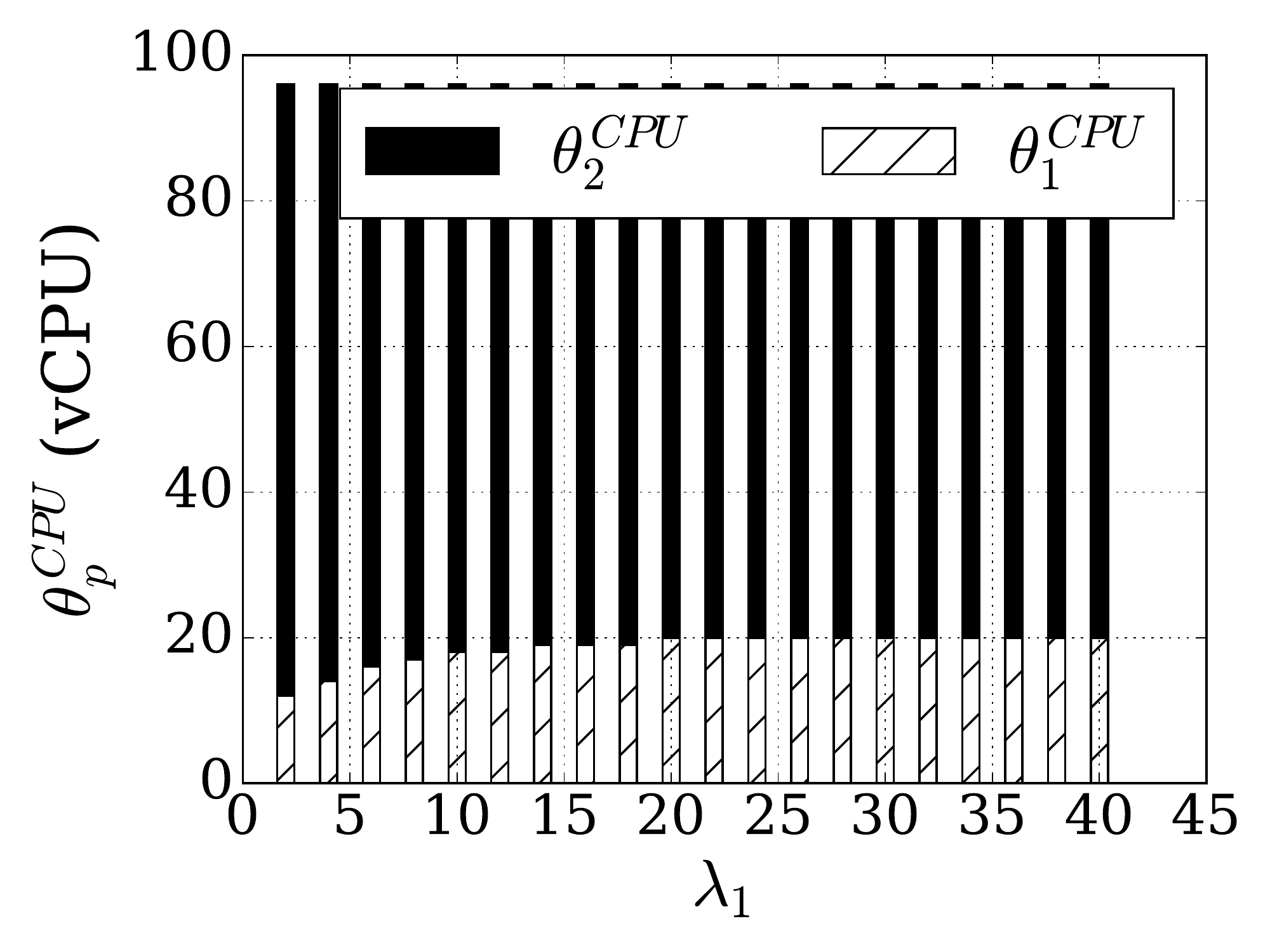}
  \caption{CPU allocation vs. $\lambda_1$}
  \label{fig:cpu-alloc}
\end{subfigure}
\caption{Resource utilization vs. $\lambda_1$}
\label{fig:alloc}
\end{figure}

\begin{figure}[t]
  \centering
  \includegraphics[width=.55\linewidth]{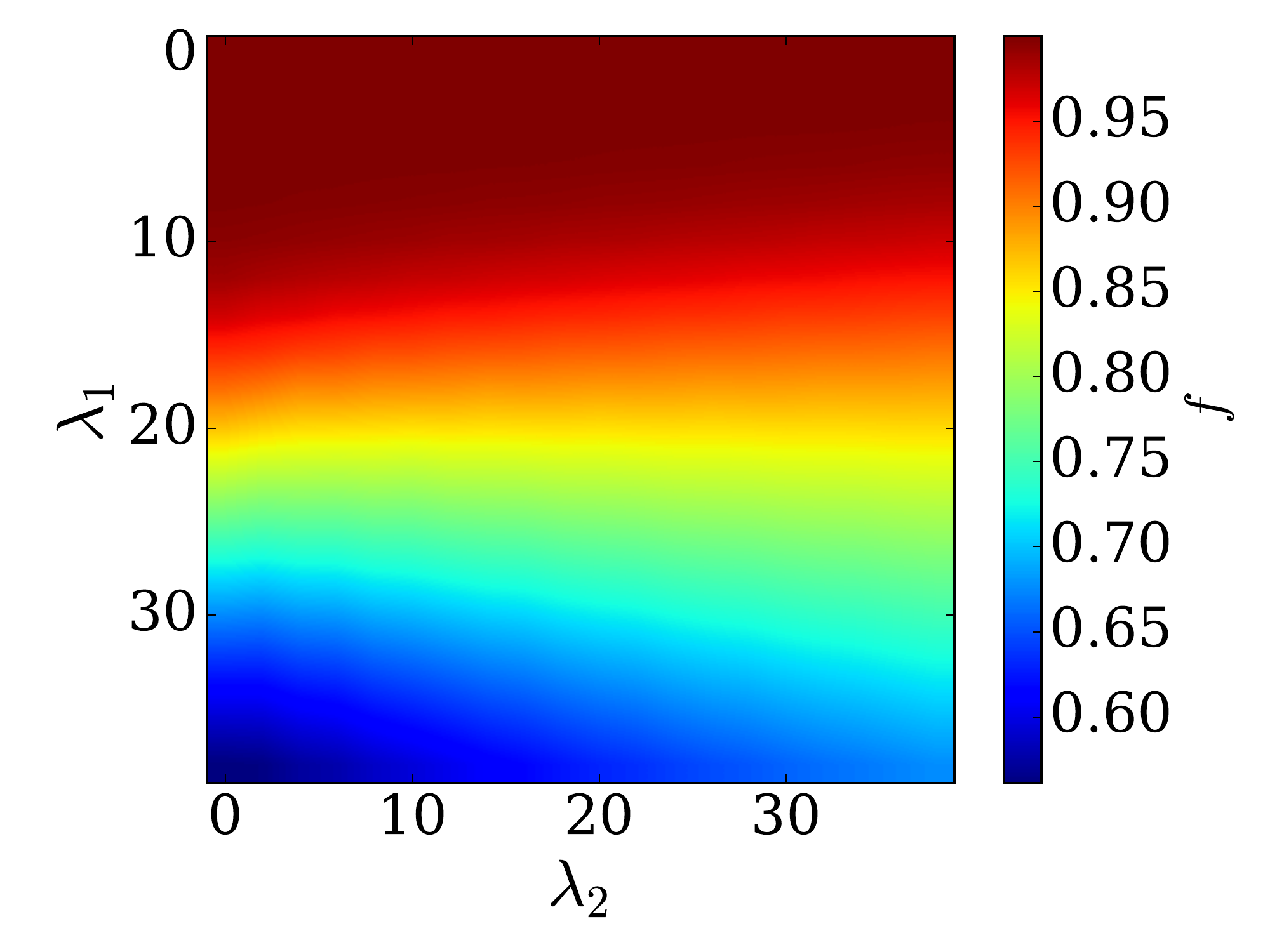}
  \caption{Objective function $f$ w.r.t $\lambda_1$ and $\lambda_2$}
  \label{fig:heat-lam}
\end{figure}

In Fig.\ref{fig:heat-lam}, we plot a heat-map describing the global objective function $f$ with respect to the variations of the two arrival rates. Obviously, the performance of the algorithm under lower arrival rates is better (dark red region $f \geq 0.95$). But what is more interesting in the figure, is that even for high arrival rates for SP 2 ($\lambda_2 \geq 35$), the algorithms keeps performing well up to $\lambda_1 = 20$ (orange region $f\geq 0.85$), no matter the arrival rate $\lambda_2$ of SP~2. The opposite is not the same: for any value of~$\lambda_2$, even small ones, the performance highly depend on $\lambda_1$. We can explain that by the fact that the users of SP 2 consume a lot of CPU (the blocking resource) which means every new admission of SP 1 would degrade the performance of the algorithm.

\begin{figure}[t]
\begin{subfigure}{.24\textwidth}
  \centering
  \includegraphics[width=1\linewidth]{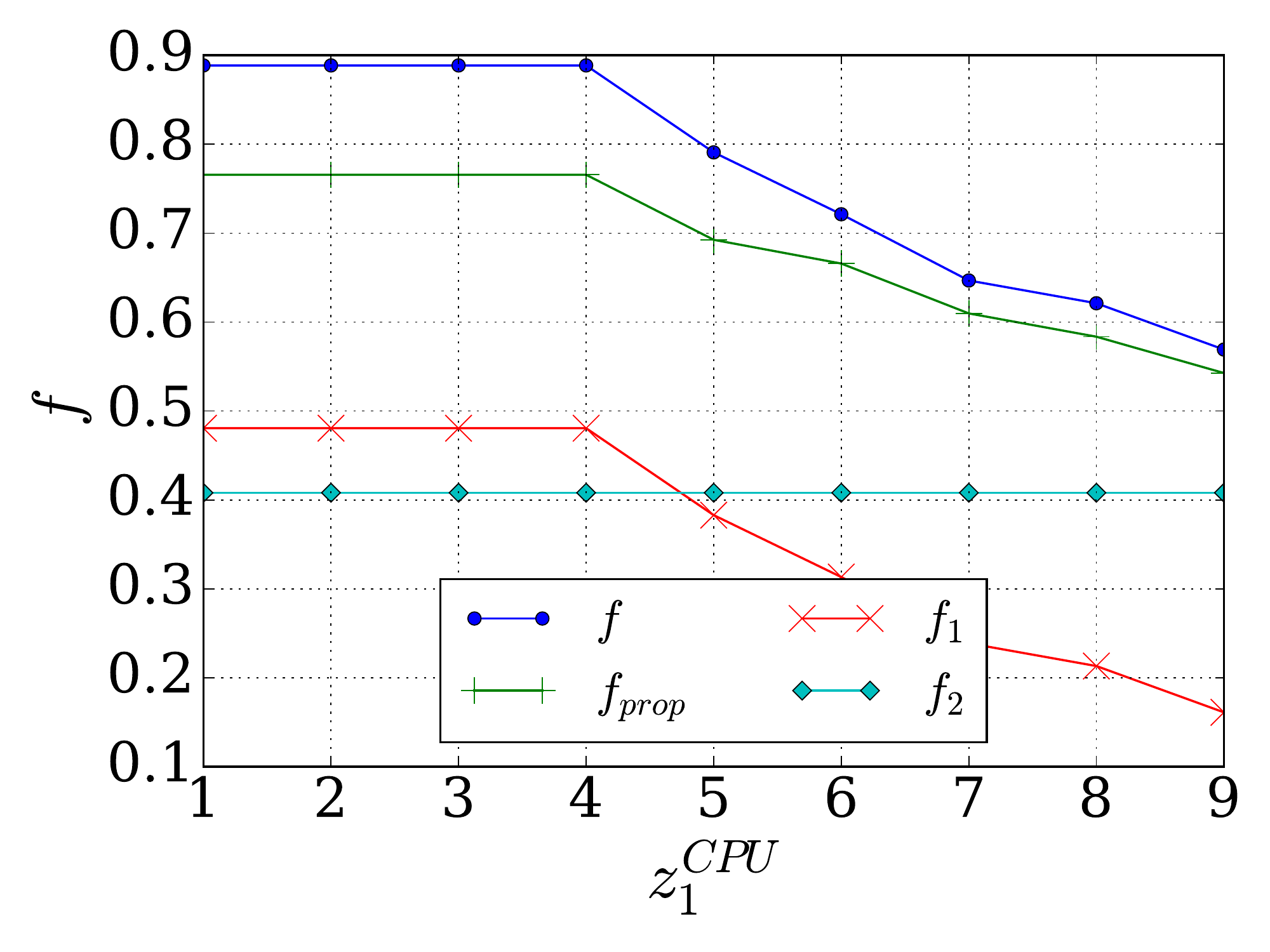}
  \caption{$f$ vs. $z_1^{\text{CPU}}$}
  \label{fig:f-vs-z1cpu}
\end{subfigure}%
\begin{subfigure}{.24\textwidth}
  \centering
  \includegraphics[width=1\linewidth]{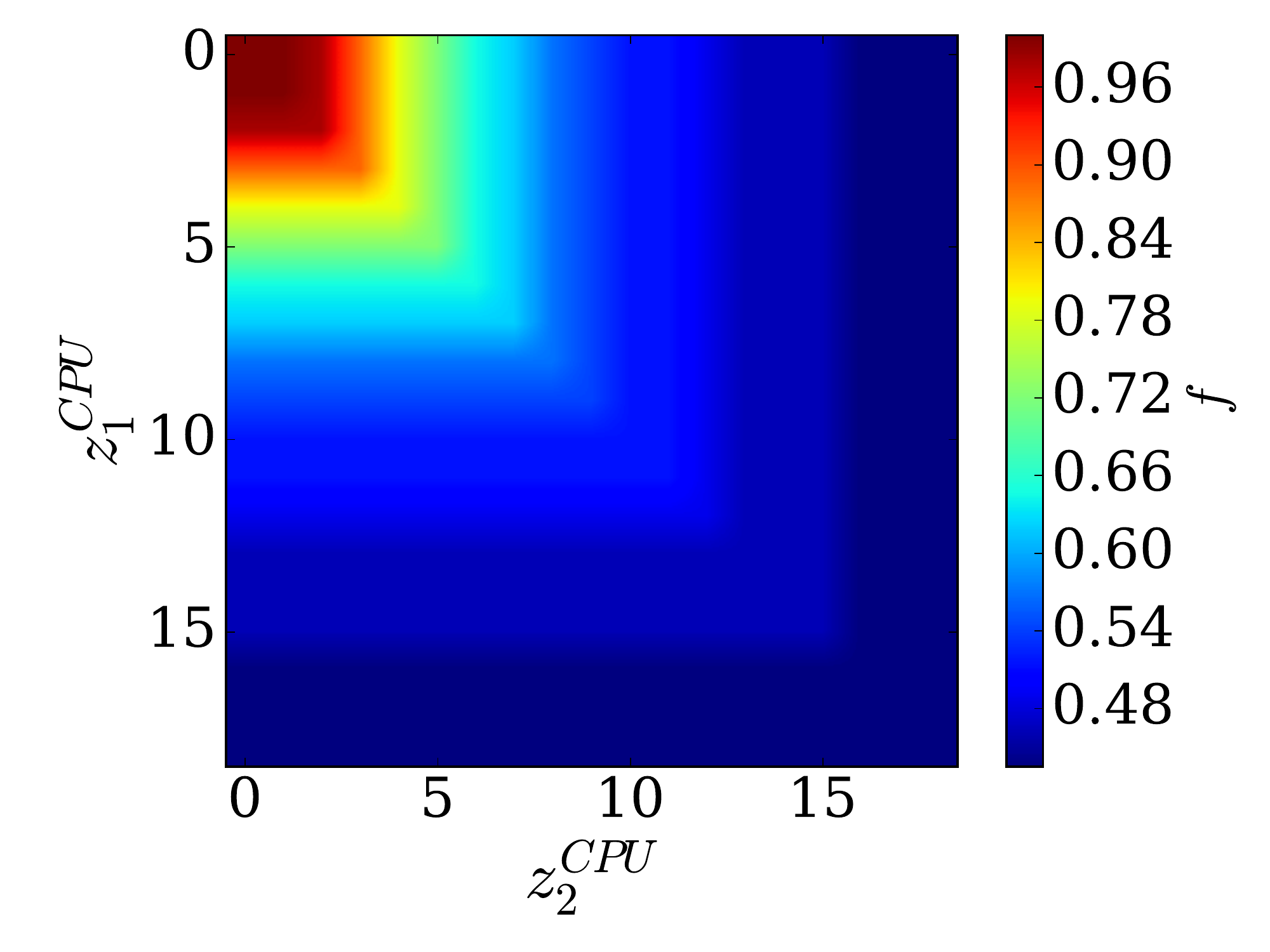}
  \caption{$f$ w.r.t $z_1^{\text{CPU}}$ and $z_2^{\text{CPU}}$}
  \label{fig:heat-zc1-zc2}
\end{subfigure}
\caption{Sensitivity w.r.t $z_p^{\text{CPU}}, p = 1, 2$}
\label{fig:sensitivity}
\end{figure}

Since the CPU is the blocking resource, we evaluate in Fig.~\ref{fig:sensitivity} the sensitivity of the system with respect to the required amount of CPU by each user of the two SPs. First, we plot in Fig.~\ref{fig:f-vs-z1cpu} the objective functions: $f$, $f_1$ and $f_2$ obtained by the algorithm and $f_{\text{prop}}$. The results show that the streaming algorithm outperforms the baseline allocation whatever users of SP 1 require in term of CPU.
In Fig.~\ref{fig:heat-zc1-zc2}, we plot the heat-map describing the global objective function $f$ obtained with the streaming algorithm with respect to the variations of the CPU requirements. The algorithm maintains a satisfying performance (dark red to light green region) up to requirements around 5 vCPU at most and then the performance rapidly decrease with the higher CPU requirements.

%%%%%%%%%%%%%% CONCLUSION AND FUTURE WORK %%%%%%%%%%%%%%%%%%%
\section{Conclusion and Future work}
\label{sec:concl}

We tackled in this paper resource allocation at EC between heterogeneous, MAR-oriented SPs competing over multiple, limited resources. We modeled the users dynamics in terms of an Erlang-type queuing model, we formulated the resource allocation problem as a sub-modular maximization problem subject to multiple knapsack constraints and solve it via an  approximation algorithm with provable optimality gap. Our numerical results quantified the performance of our algorithm in terms of the probability that users get served by the Edge, as opposed to being blocked and re-directed towards the Cloud which entails larger delay and hence lesser QoE. We showed the resulting resources partitioning between the SPs. 
We showed the algorithm outperforms a baseline resource allocation, proportional to users arrival rates. Finally, we included a sensitivity analysis with respect to individual user requirement of a given resource. Our next work perspective would focus on the case where users arrival rates as well as resource requirements are unknown, the NO shall implement learning in order to be able to allocate resources in this case. 

%that the proposed algorithm outperforms baselines allocations

\bibliographystyle{IEEEtran}
\bibliography{main}

\end{document}